\documentclass[11pt]{article}
\usepackage[english]{babel}
\usepackage{amssymb}
\usepackage{amsthm}
\usepackage{amsmath}
\usepackage{a4wide}
\usepackage{esvect}
\usepackage{hyperref}
\usepackage{verbatim}
\usepackage[dvips]{graphicx}
\usepackage{color}
\usepackage{geometry}
\usepackage{framed}
\newtheorem{thm}{Theorem}

\newtheorem{lemma}[thm]{Lemma}

\newtheorem{prop}[thm]{Proposition}

\newtheorem{defn}[thm]{Definition}
\newtheorem{definition}[thm]{Definition}

\newtheorem{fact}[thm]{Fact}
\newtheorem{cor}[thm]{Corollary}

\newtheorem{rem}[thm]{Remark}
\newtheorem{clm}[thm]{Claim}

\newcommand\cA{{\mathcal A}}

\newcommand\cF{{\mathcal F}}
\newcommand\cG{{\mathcal G}}

\newcommand{\ignore}[1]{}

\title{Smart elements in combinatorial group testing problems}

\begin{document}

\author{
D\'aniel Gerbner\thanks{Research supported by the J\'anos Bolyai Research Fellowship of the Hungarian Academy of Sciences.} \thanks{Research supported by the National Research,
Development and Innovation Office -- NKFIH, grant K116769.}
\and
M\'at\'e Vizer\thanks{Research supported by the National Research, Development and Innovation
Office -- NKFIH, grant SNN 116095.}}

\date{MTA R\'enyi Institute\\
Hungary H-1053, Budapest, Re\'altanoda utca 13-15.\\
\small \texttt{gerbner@renyi.hu, vizermate@gmail.com}\\
\today}

\maketitle

\begin{abstract}
In combinatorial group testing problems the questioner needs to find a special element $x \in [n]$ by testing subsets of $[n]$. Tapolcai et al.~\cite{TRHHS2014,TRHGHS2016} introduced a new model, where each element knows the answer for those queries that contain it and each element should be able to identify the special one. 

Using classical results of extremal set theory we prove that if $\mathcal{F}_n \subseteq 2^{[n]}$ solves the non-adaptive version of this problem and has minimal cardinality, then
$$\lim_{n \rightarrow \infty} \frac{|\mathcal{F}_n|}{\log_2 n} = \log_{(3/2)}2.$$ This improves results in~\cite{TRHHS2014,TRHGHS2016}. 

We also consider related models inspired by secret sharing models, where the elements should share information among them to find out the special one. Finally the adaptive versions of the different models are investigated.

\end{abstract}

\section{Introduction}

\subsection{Classical model, basic definitions}

In the most basic \textit{model} of \textit{combinatorial group testing} the questioner (we call him Questioner in the following) needs to find a special element $x \in [n](:= \{1,2,\dots,n \})$. He can test subsets of $[n]$ and for $F\subseteq [n]$ the answer is the appropriate value of the function $t: 2^{[n]} \rightarrow \{NO,YES\}$ defined by:
$$t(F):=
\left\{ \begin{array}{l l}
YES & \textrm{ if } x \in F,\\
NO & \textrm{ if } x \not \in F.\\
\end{array}
\right.$$
The tested subsets are called \textit{queries} and the special element is usually called \textit{defective} in the group testing literature. Questioner's aim is to ask as few queries as possible and the number of queries needed to ask in the worst case is called the \textit{worst-case complexity} of the problem. For any combinatorial group testing problem there are at least two main approaches: whether it is \textit{adaptive} or \textit{non-adaptive}. In the adaptive scenario Questioner asks queries depending on the answers for the previously asked queries, however in the non-adaptive version Questioner must pose all the queries at the beginning. 

Let us briefly describe the solution for the (above mentioned) most basic combinatorial group testing model in the non-adaptive case. 
We call a family $\mathcal{F} \subseteq 2^{[n]}$ \textbf{separating} if for any two different $x,y \in [n]$ there is $F \in \mathcal{F}$ with $x\in F \textrm{ and } y \not\in F, \textrm{ or } y\in F \textrm{ and } x \not\in F.$ 

\begin{fact} 
Questioner finds the defective by asking elements of $\mathcal{F} \subseteq 2^{[n]}$ if and only if $\mathcal{F}$ is separating.
\end{fact}

The notion of \textit{separating family} in the context of combinatorial group testing was introduced and first studied by R\'enyi in \cite{R1961}. We will also use the following simple fact later:

\begin{fact}

Suppose $\mathcal{F}_n \subseteq 2^{[n]}$ is the smallest separating family. Then we have:
$$|\mathcal{F}_n| = \lceil \log_2 n \rceil.$$

\end{fact}

One can imagine many possible generalizations of the most basic classical model: more defectives, other answers (threshold \cite{D2006} or density \cite{GKPW2013} group testing), average case complexity \cite{CHKV2011}, rounds \cite{DMT2013,GV2016}. For a survey on different non-adaptive models see e.g. \cite{DH2006}.

Combinatorial group testing problems were first considered during the World War II by Dorfman \cite{D1943} in the context of mass blood testing. Since then group testing techniques have had many different applications, for example in fault diagnosis in optical networks  \cite{HPWYC2007}, in quality control in product testing
\cite{SG1959} or failure detection in wireless sensor networks \cite{LLLG2013}.

\subsection{New feature of the elements}

Inspired by the \textit{node failure localization model} of Tapolcai et al.~\cite{TRHHS2014,TRHGHS2016} we introduce a possible new feature of the elements. Informally speaking an element can be kind of smart and this fact means two things: 

\vspace{2mm}

1) it knows the answer to those queries that contain it, and 

\vspace{1mm}

2) it can deduce information from the results of the tests it is involved.

\vspace{3mm}

\noindent
Let us define these properties more formally and introduce our main definitions.

\begin{definition} \

$\bullet_1$ We say that an element $x \in [n]$ is \textbf{smart}, if for any set of queries $\mathcal{F} \subseteq 2^{[n]}$ $x$ is aware of the answers for the queries $\cF_x:=\{F\in \cF: x\in F\}$.

$\bullet_2$ We say that a smart element $x$ \textbf{knows the defective} element, if the asked query family $\mathcal{F}$ satisfies the following property: no matter what the defective element is, after the answers $x$ can find the defective one, or equivalently the subfamily $\cF_x$ is separating. 

$\bullet_3$ We say that a smart element $x$ \textbf{does not know the defective element}, if the query family satisfies the following property: no matter what the defective element $y$ is, after the answers $x$ does not know that $y$ is the defective, or equivalently for any $y \in [n]$ there is a different $z \in [n]$ that is contained in exactly the same members of $\cF_x$ as $y$. 

\end{definition}

\vspace{2mm}

Note that the above two cases ($\bullet_2$ and $\bullet_3$) do not cover all the possibilities: if $x$ is contained only in the sets $\{x,y\}$ and $\{x,z\}$ with $n \ge 5$, then we cannot say that $x$ does not know the defective, but we also cannot say that $x$ knows the defective. Indeed, if the defective is $x$, $y$ or $z$, then $x$ knows, while otherwise $x$ does not know the defective.

\

\subsection{Possible applications of smart elements}

One can imagine many situation, where the tested items have computational capacity, so they can become 'smart'. We list some scenarios, where these elements can be used:

\

$\bullet$ \textbf{Find the defective and distribute information among elements}. Let us suppose that we have a wireless router/mobile network, or just a system of smart devices and one of them becomes faulty. We want to find it by testing (any) subsets of the elements of the network and also share this information with every (other) element to prevent sending information to the disabled unit. However, the smart devices actively participate in the tests they are involved in, thus they might be able to see the results of those tests. In that case it is useful if they can identify the disabled unit without further communication.
Another advantage is that they do not need a 'chief', who conducts the whole procedure. We will asymptotically determine the number of tests needed to solve this problem later in this article.

A version of the previously mentioned problem  already appeared in the literature. In \cite{TRHHS2014,TRHGHS2016} failures in a network are checked by monitoring trails that turn into off state if interrupted by a failure event. The goal is to construct the monitoring trails such a way that any node can determine the network failure status solely by observing the on-off status of the monitoring trails traversing that node. The network is given by a graph and the monitoring trails are subgraphs satisfying certain properties. This is the same model with the additional assumption that we cannot test any subsets, only some special ones. 

However, the lower bound proved in \cite{TRHHS2014,TRHGHS2016} does not use this property, but deals with the abstract setting studied in this paper. We improve their lower bound in Corollary \ref{cormodel3}.

\

$\bullet$ \textbf{Distrust Questioner.} We mention another motivation of our investigations: it is often mentioned in the group testing literature that an advantage of testing pools together is that it increases privacy.
Assume that the tested elements (that can be people, computers etc.) distrust Questioner, thus they want to control the tests they are involved in, and as a consequence they will know the answer for these tests. However in this case we might not want that the tested elements could find out which one of them is the defective, because of privacy reasons . The systematic research on this property has only started recently, see e.g. \cite{AFBC2008,CCG2016,EGH2013}, however these papers focus rather on cryptographic versions of the problem. 

Here we deal with a simple combinatorial version, where privacy only means that an unauthorized participant cannot completely detect the defective element. Note that if each element knows there is exactly one defective, every query immediately shows several elements which are not defective - either the elements of the test, or the ones in the complement. As we do not use any encryption, the elements of that set gain significant information. This is why we can only require that elements cannot completely detect the defective one.
They might be able to narrow it down to two candidates, but cannot completely identify it.

\subsection{Another new feature}

It is possible that in some model some elements can not identify the defective, however if we pick two elements and they share their information among them, they can find the defective element. Hence motivated by secret sharing schemes (see e.g. \cite{B2011}), in some models we also consider the following new feature of the smart elements: they can work together and share their knowledge. More formally a set of smart elements $X \subseteq [n]$ \textbf{shares their knowledge among them}, if all elements will know the answers for all the queries in $\bigcup_{x \in X} \cF_x$. But we emphasize that we do not deal with the way the data is transmitted. Information can not be distributed between different groups. Elements will have this feature just in Model 4.

\

\textbf{Structure of the paper.} We organize the paper as follows: in Section 2 we introduce some properties and related results about families of sets, that we will need later. In Section 3 we give a general introduction of the investigated models and state a result about Model 1. In Section 4 we introduce Model 2, then state and prove our related results. In Section 5 we continue with Model 3 (and its variants), while we finish the investigation about the non-adaptive models with Model 4. In section 6 we focus on the possible adaptive scenarios. We finish this article with some remarks and open problems in Section 7.

\section{Finite set theory background}

In our proofs we will use the language of (extremal) finite set theory. In this section we introduce some notions on families of subsets and known results about them, that we will use. First some general ones:

\vspace{2mm}

The \textbf{complement} of a family $\cF\subseteq 2^{n}$ is $\cF^{c}:=\{[n] \setminus F: F \in \cF\}$, while the \textbf{dual} of it is $\cF':=\{\cF_a: a\in [n]\}$ (recall that $\cF_a=\{F\in \cF: a\in F\}$). Note that $\cF'$ is defined on the underlying set $\cF$ and has cardinality at most $n$.

\vspace{3mm}

Now we introduce some more specific notions about families of subsets of $[n]$.

\begin{defn} \indent We say that $\cF \subseteq 2^{[n]}$ is:

\vspace{2mm}

$\bullet_1$ \textbf{intersection closed} if $F,G\in \cF$ implies $F\cap G\in\cF$.

\vspace{2mm}

$\bullet_2$ \textbf{Sperner} if there are no two different $F_1,F_2 \in \mathcal{F}$ with $F_1 \subseteq F_2$.

\vspace{2mm}

$\bullet_3$ \textbf{cancellative} if for any three $F_1,F_2,F_3 \in \mathcal{F}$ we have $$F_1 \cup F_2 = F_1 \cup F_3 \Rightarrow F_2=F_3.$$

$\bullet_4$ \textbf{intersection cancellative} if for any three $F_1,F_2,F_3 \in \mathcal{F}$ we have $$F_1 \cap F_2 = F_1 \cap F_3 \Rightarrow F_2=F_3.$$

$\bullet_6$ \textbf{completely separating} if for any two different $x,y \in [n]$ there is $F \in \mathcal{F}$ with $$ x \in F \textrm{ and } y \not \in F.$$

$\bullet_7$ a \textbf{pairwise balanced design} if for every two different elements $x,y\in [n]$ there is 

\hspace{4mm}exactly one $F\in\cF$ that contains both. If $K$ is the set 
of cardinalities of the members 

\hspace{4mm}of $\cF$, we say $\cF$ is a \textbf{PBD}($K$). If $K=\{3\}$, we say $\cF$ is a \textbf{Steiner triple system}.

\end{defn}

\subsubsection*{Some known results about these notions that we will use later}

\noindent $\bullet$ The notion \textit{cancellative} was introduced by Frankl and F\"uredi in \cite{FF84}, where they proved the following upper bound on the size of a cancellative family of subsets:

\begin{thm}\label{ffcan}(Frankl, F\"uredi \cite{FF84}, Theorem 3)

Suppose that $n \ge 14$ and $\mathcal{F} \subseteq 2^{[n]}$ is cancellative. Then we have $$|\mathcal{F}| \le n \cdot (\frac{3}{2})^n.$$

\end{thm}

The following theorem was proved by Tolhuizen:

\begin{thm}\label{tcan}(Tolhuizen \cite{T00}, Corollary 1)

Suppose $\mathcal{F}_n \subseteq 2^{[n]}$ is the largest cancellative family, then we have:
$$\lim_{n \rightarrow \infty} \frac{1}{n} \log_2 |\mathcal{F}_n|=\log_2 (\frac{3}{2}).$$

\end{thm}

\noindent
We will also use the following during the proof of our results:

\begin{fact}\label{intcan}

$\mathcal{F} \subseteq 2^{[n]}$ is intersection cancellative if and only if $\mathcal{F}^c=\{[n] \setminus F : F \in \mathcal{F} \}$ is cancellative.

\end{fact}

\vspace{2mm}

\noindent
$\bullet$ The notion of \textit{completely separating} family was introduced by Dickson in \cite{D1969}, where he determined the order of the smallest completely separating family. Later Spencer observed the following:

\begin{thm}\label{compsep} (Spencer, \cite{S1970})
For  $\mathcal{F} \subseteq 2^{[n]}$ ($n \ge 1$) is completely separating if and only if its dual is Sperner.
Thus for any $n\ge 1$ there exists a completely separating family $\mathcal{F}_n \subseteq 2^{[n]}$with:
$$|\mathcal{F}_n| \le \lceil \log_2 n + \frac{1}{2} \log_2 \log_2 n \rceil.$$

\end{thm}

\vspace{3mm}

\noindent $\bullet$ The notion of \textit{Steiner triple systems} was introduced in the middle of the 19th century and has since developed into the huge area of combinatorial designs. Here we will use two of the most fundamental results. A subfamily of pairwise disjoint sets is a \textit{partial matching}, and it is a \textit{matching} if it covers all the elements. They are also called parallel classes in design theory.

\begin{thm}\label{steiner} (Kirkman \cite{K1847}, Bose \cite{B1939}, Skolem \cite{S1958}) There exists a Steiner triple system on $[n]$ if and only if $n=6k+1$ or $n=6k+3$ for some integer $k$.

\end{thm}

\begin{thm} (Ray-Chaudhuri, Wilson \cite{RW1971}) If $n=6k+3$, then there exists a Steiner triple system that can be decomposed into $3k+1$ complete matchings.

\end{thm}

\section{General introduction to the models and Model 1}

In this section we give a general introduction to our models and start our investigations.

\medskip

In all our models we have:

$\bullet$ an input set of $n$ smart elements, and one of them is defective. 

$\bullet$ Model 1-4 are non-adaptive models, so Questioner needs to construct a family $\cF \subseteq 2^{[n]}$ of tests at the beginning. 

$\bullet$ A test is a subset $F\subseteq [n]$ corresponding to a query of the following type: 'is the defective an element of $F$?', and the answer is NO if $F$ does not contain the defective and YES, if it contains the defective. 

$\bullet$ As we mentioned all the elements are smart elements in all the models, so for a test $F$ every element of $F$ knows the answer in addition to Questioner. 

$\bullet$ In each model we assume that knowing all the answers is enough information for Questioner to find the defective element, i.e. $\cF$ is separating. 

$\bullet$ The main difference between Model 1-4 is what we want the elements to find out. Using only the information available to them, i.e. the answers to the queries containing them, we can require that they find out something about the defective element, or oppositely, that they cannot find out something. We will indicate the aim as the \textbf{property} of a certain model.

$\bullet$ We say that $\cF \subseteq 2^{[n]}$ \textbf{solves} that model if the property of the model is reached by asking elements of $\cF$.

\vspace{2mm}

In each of the following models we first give a property describing  what the elements should know, and then we examine if there is a query family that solves that specific model or state results about the cardinality of such query families. First we consider the models where we require the elements to find out something about the defective (like the model by Tapolcai et al.~\cite{TRHHS2014,TRHGHS2016} that initiated the research is of this type). Then we consider the models where we require some information to remain hidden from the elements. Finally we mix these types of properties in Model 4.

\subsection{Model 1}

The most natural model is the following:

\vspace{2mm}

\textbf{Property:} all elements know (each about itself) if they are defective.

\vspace{2mm}

\noindent
It is easy to see that this property is equivalent to the following: for every two different $x,y \in [n]$ there is a set $F\in\cF$ such that $x\in F$, $y\not\in F$, i.e. $\cF$ is completely separating. By Theorem \ref{compsep} we immediately have:

\vspace{2mm}

\begin{prop}

For any $n \ge 1$ there is $\mathcal{F}_n \subseteq 2^{[n]}$ that solves Model 1 with:
$$|\mathcal{F}_n| \le \lceil \log_2 n + \frac{1}{2} \log_2 \log_2 n \rceil.$$

\end{prop}

\section{Model 2}

This model is the abstract version of the node failure localization model introduced by Tapolcai et al.~\cite{TRHHS2014,TRHGHS2016}.

\vspace{2mm}

\textbf{Property:} all elements know the defective.

\vspace{2mm}

\noindent
Lenger \cite{L2016} proved that there is $\mathcal{F}_n \subseteq 2^{[n]}$ that solves Model 2 with $|\mathcal{F}_n| \le 3 \log_3 n$ (that is a better upper bound than the ones in~\cite{TRHHS2014,TRHGHS2016}. However we note again the latter results are about non-abstract cases.). In the following (see Corollary \ref{cormodel3}) we prove an asymptotically sharp result on the minimal cardinality of the solutions of Model 2. 

To reach that result first we characterize the query families that solve Model 2. 

\begin{thm}\label{model3} $\cF_n \subseteq 2^{[n]}$ solves Model 2 if and only if its dual is Sperner and intersection-cancellative.

\end{thm}

\begin{proof}[Proof of Theorem \ref{model3}]

We start the proof with the following easy lemma that gives some characterization of the query families that solve Model 2. 

But before that we introduce the following notion: we say $x$ \textit{distinguishes} between $y$ and $z$ if in case $y$ or $z$ is the defective, $x$ can tell which one it is, using the answers to the queries containing $x$. Equivalently, there is a query that contains $x$ and exactly one of $y$ and $z$.

\begin{lemma}\label{char}

$\mathcal{F} \subseteq 2^{[n]}$ solves Model 2 if and only if the following two properties hold:

\vspace{2mm}

$\bullet_1$ $\mathcal{F}$ is completely separating, and 

\vspace{2mm}

$\bullet_2$ for all pairwise different $a,b,c \in [n]$ there is $F \in \mathcal{F}$ with $a,b \in F$ and $c \not \in F$ or with 

\hspace{4mm} $a,c \in F$ and $b \not \in F$.

\end{lemma}

\begin{proof}[Proof of Lemma \ref{char}]

We prove by contradiction. 

\vspace{2mm}

First suppose that $\bullet_1$ is not true. So there are two different elements $a,b \in [n]$ such that for all $F \in \mathcal {F}$ if $a \in F$, then $b \in F$. In this case Adversary answers YES for all queries that contain $a$ and $a$ will not be able to distinguish $a$ and $b$ and decide whether $a$ or $b$ is the defective.

If $\bullet_2$ is not true, then there are three different $a,b,c \in [n]$ such that for all $F \in \mathcal{F}$ if $a,b \in F$, then $c \in F$ and if $a,c \in F$, then $b \in F$. If Adversary answers YES for all queries that contain $a,b$ and $c$, then $a$ will not be able to decide whether $b$ or $c$ is the defective.

\vspace{2mm}

To prove the other direction first observe that by $\bullet_1$ only the defective element gets YES answer for all the queries containing it. Thus any other element knows that he is not a defective (getting at least one NO answer (for a query containing it)). However by $\bullet_2$ he can decide who is the defective. Indeed he can consider the intersection of all the queries that were answered YES and contained him. There is exactly one other element in the intersection, and that is the defective).

\end{proof}

Now we translate the properties of $\mathcal{F}$ given in Lemma \ref{char} for the properties of the dual of $\mathcal{F}$.

\begin{lemma}\label{dual}

$\mathcal{F} \subseteq 2^{[n]}$ satisfies properties $\bullet_1$ and $\bullet_2$ if and only if its dual is Sperner and intersection cancellative.

\end{lemma}

\begin{proof} The fact that the dual of a completely separating system (property $\bullet_1$ of Lemma \ref{char}) is Sperner was proved in \cite{S1970}  (as we mentioned it earlier in Theorem \ref{compsep}).

Therefore it is enough to prove that the dual of a family with property $\bullet_2$ of Lemma \ref{char} is cancellative. Property $\bullet_2$ means that for any three different sets $A,B,C$ in the dual there is an element $f$ (corresponding to $F$) such that either $f\in A$, $f\in B$ and $f\not\in C$ or $f\in A$, $f\in C$ and $f\not\in B$. This means either $f\in A\cap B\setminus C$ or $f\in A\cap C\setminus B$. The existence of $f$ means either $A \cap B \not \subseteq C$ or $A \cap C \not \subseteq B$. Let us define three properties.

\vspace{1mm}
$\circ_1$ $A \cap B \not \subseteq C$.

\vspace{1mm}

$\circ_2$ $A \cap C \not \subseteq B$.

\vspace{1mm}

$\circ_3$ $C \cap B \not \subseteq A$.

\vspace{1mm}

Property $\bullet_2$ (for these three sets in this order) means that at least one of $\circ_1$ and $\circ_2$ holds. Considering the same three sets in different orders we get that also at least one of $\circ_1$ and $\circ_3$ and one of $\circ_3$ and $\circ_2$ holds. It is true if and only if at least two of these three properties hold.

\vspace{2mm}

\noindent
To finish the proof of Lemma \ref{char} we prove the following:

\begin{clm}\label{canc}

$\mathcal{F}' \subseteq 2^{[n]}$ is intersection cancellative if and only if at least two out of $\circ_1, \circ_2$ and $\circ_3$ hold for any three members of it. 

\end{clm}

\begin{proof}

Let us assume $\cF'$ is intersection cancellative and let $A,B,C \in \mathcal{F}'$. Let us assume at most one, say $\circ_3$ of the three properties holds, thus $\circ_1$ and $\circ_2$ do not hold. The first one implies $A\cap B \subseteq C$, and obviously $A\cap B\subseteq A$. Thus we have $A\cap B \subseteq A\cap C$. Similarly the second one implies $A\cap C \subseteq A\cap B$, hence they together imply $A\cap C=A\cap B$, which contradicts the intersection cancellative property and our assumption that $A,B,C$ are three different sets.

Let us assume now that $\cF'$ is not intersection cancellative, thus we have $A\cap B=A\cap C$. This implies both $A\cap B \subseteq C$ and $A\cap C\subseteq B$, thus at most one of  $\circ_1, \circ_2$ and $\circ_3$ can hold.

\end{proof}

We are done with the proof of Lemma \ref{dual}.

\end{proof}

By Lemma \ref{char} and Lemma \ref{dual} we are done with the proof of Theorem \ref{model3}.

\end{proof}

\noindent With the help of the previous theorem we can prove the following:

\begin{cor}\label{cormodel3} Suppose $\mathcal{F}_n \subseteq 2^{n}$ solves Model 2 and has minimal cardinality. Then we have
$$\lim_{n \rightarrow \infty} \frac{|\mathcal{F}_n|}{\log_2 n} = \log_{(3/2)}2 \ (\approx 1.70951).$$

\end{cor}

\begin{rem}
This result provides an improvement of the results of Theorem 1 of \cite{TRHHS2014} and \cite{TRHGHS2016}.Tapolcai et al.~\cite{TRHHS2014,TRHGHS2016} proved that $1.62088\log_2 n$ queries are needed in the abstract setting, and gave examples of graphs where $2\log_2 n$ monitoring trails are needed. Here we improve their lower bound, and show that at least $\log_{(3/2)}2 \log_2n\ge 1.70951\log_2 n$ queries are needed, and this bound is asymptotically sharp in the abstract case.

\end{rem}

\begin{proof}[Proof of Corollary \ref{cormodel3}]

First note that by Theorem \ref{model3} and Fact \ref{intcan} we have that $\cF_n \subseteq 2^{[n]}$ solves Model 2 if and only if the complement of its dual is Sperner and cancellative. Now the upper bound

$$\limsup_{n \rightarrow \infty} \frac{|\mathcal{F}_n|}{\log_2 n} \le \log_{(3/2)}2$$
follows from Theorem \ref{ffcan}. Note that we do not use that $\cF_n$ is also Sperner.

\vspace{2mm}

Now we start to work towards the lower bound. Theorem \ref{tcan} gives a large (not necessarily Sperner) cancellative family. However, a more careful analysis of Tolhuizen's proof \cite{T00} shows that the family given there is Sperner. We just give a sketch here as it introduces a lot of new definitions. 

A set $X\subseteq [n]$ is an \textit{identifying set} for a family $\cG\subseteq 2^{[n]}$  if for any members $G,G'\in\cG$ there exists $x\in X$ such that either $x\in G\setminus G'$ or $x\in G'\setminus G$. Tolhuizen proved that for any family $\cG$ the family of sets that are both members of $\cG$ and identifying sets for $\cG$ is intersection cancellative. To get a large intersection cancellative family he used codes and constructed a family $\cG$ that contained many sets that were also identifying sets for $\cG$. Observe that if $A\subseteq B$ with $A,B\in \cG$, then $A$ cannot be an identifying set, as elements of it cannot be in $A\setminus B$ nor in $B\setminus A$. This implies the resulting intersection cancellative family is also Sperner. Thus we have

$$\liminf_{n \rightarrow \infty} \frac{|\mathcal{F}_n|}{\log_2 n} \ge \log_{(3/2)}2.$$

\vspace{2mm}

We saw that Tolhuizen's construction is Sperner, however we note that even starting from a large cancellative family that is not Sperner, we could consider the largest subfamily of it that consists of sets of the same size. The resulting Sperner family would still be large enough to give the same asymptotic result.

\end{proof}

\section{Model 3}

In this model Questioner wants to find the defective such a way that its identity is hidden from the participants themselves.

\vspace{2mm}

\textbf{Property}: no element knows the defective.

\begin{prop}

No $\mathcal{F}$ can solve Model 3.

\end{prop}

\begin{proof}
Recall that we always assume that Questioner can find the defective, i.e.~$\cF$ is separating. Let us consider the families $\cF_x$ ($x \in [n]$) and choose an element $x$ such that $\cF_x$ is inclusion-wise maximal among these families. We claim that if $x$ is the defective, then he knows that. Indeed, $x$ gets only YES answers. Suppose by contradiction that $y$ could also be the defective according to $x$, then we would have $\cF_y\supseteq \cF_x$, which implies $\cF_y= \cF_x$. However it is impossible, as $\cF$ is separating.

\end{proof}

\subsection{Model 3'}

As Model 3 is impossible to solve, in the next model the defective himself may find out he is the defective, but nobody else (note that we assume that knowing all the answers is enough to find the defective).

\vspace{2mm}

\textbf{Property:} no element knows the defective, except for the defective one.

\vspace{2mm}

\noindent
Opposed to Model 3, this is easily achievable: we can ask all (or all but one) of the singletons. So a natural question that arises here is the cardinality of the smallest family that can solve Model 3'. In the next theorem we give an upper bound on this quantity.

\begin{thm}\label{model3'}

For every $n\ge 1$ there is $\cF_n\subseteq 2^n$ that solves Model 3' with $$|\cF_n |\le 3\lceil\log_3 n \rceil -t(n),$$ where $t(n)$ is the number of zeros in $n$ written in ternary base.
\end{thm}

\begin{proof}[Proof of Theorem \ref{model3'}]

We construct $\cF_n$ recursively. If $n\le 8$, then it is easy to check that there is $\cF_n$ that solves Model 3' and $|\cF_n| \le 3\lceil\log_3 n \rceil -t(n)$.

Let us assume $n \ge 9$ and consider a family $\cF$ that solves Model 3' on $\lfloor n/3\rfloor$ elements. Let us replace each element $x$ by a set $A_x$ of three or four new elements to get $n$ elements altogether. For every set $F\in\cF$ let $A_F=\cup_{x\in F} A_x$. Let us also consider three disjoint sets $B_1,B_2,B_3$ such that $|A_x\cap B_i|=1$ for every $x\in [\lfloor n/3\rfloor]$ and $i=1,2,3$. Let $\cA=\{A_F:F\in \cF\}\cup \{B_1,B_2,B_3\}$ and $\cA_0=\{A_x:x\in [\lfloor n/3\rfloor]\}\cup \{B_1,B_2\}$. 

\begin{clm}\label{clmmodel3'}
$\cA$ solves Model 3' if $3\nmid n$ and $\cA_0$ solves Model 3' if $3\mid n$. 
\end{clm}

\begin{proof}[Proof of Claim \ref{clmmodel3'}]
First we prove that both $\cA$ and $\cA_0$ satisfy the property of Model 3'. Indeed, let $y\in[n]$. Let us first forget about the queries $B_1,B_2$ (and $B_3$) and consider the remaining queries. By the construction of the remaining queries if (from that information): 
 $y$ can find out which one of the sets $A_x$ contains the defective, then $A_x$ contains the defective and $y \in A_x$. However in this case $y$ (if $y$ is not the defective) cannot distinguish the other elements of $A_x$, even using the answer for the $B_i$ that contains it.

On the other hand if $y$ can not decide (again, without the $B_i$'s) which $A_x$ contains the defective, then there are at least two sets $A_x,A_z$ such that he cannot tell which one contains the defective element. Then without the sets $B_i$ he cannot distinguish them at all, thus all the (at least) $6$ elements of $A_x$ and $A_z$ should be considered as possible defective by $y$. However there is at most one $B_i$ that $y$ can use, and it intersects these (at least) two sets in (at least) two elements. Thus $y$ cannot distinguish these (at least) two elements from each other, nor the other at least $4$ elements from each other.

\vspace{2mm}

Finally we prove that both $\cA$ and $\cA_0$ are separating: if two elements are in different $A_x$, they are separated by the queries $A_F$. If they are in the same $A_x$, they are separated by $B_1,B_2,B_3$, or if $|A_x|=3$, then by $B_1,B_2$. We are done with the proof of Claim \ref{clmmodel3'} as if $n$ is divisible by three, then every $A_x$ has size $3$.

\end{proof}

By Claim \ref{clmmodel3'} we are done with the proof of Theorem \ref{model3'} as during this process we have a number divisible by three every time there is a 0 in the ternary form of $n$.

\end{proof}

\section{Model 4}

Now we start to investigate models where elements can share information among them. When we say that a group of elements together knows the defective element, we mean that all of them in the group know the answers for the queries that contained at least one of them, and using this information they can find the defective one. (Recall that information can not be distributed between different groups.) Let $i$ and $j$ be integers with $1 \le i < j \le n$.

\vspace{2mm}

\textbf{Property:} any $j$ elements together know the defective, but $i$ elements together do not know, unless one of them is the defective itself. 

\vspace{2mm}

Note that $i=0$ is another possibility. In that case the solution would be a family where any $j$ elements together can find the defective. However, in this section we only deal with the existence of a solution, and a solution for Model 2 is obviously a solution for this model as well. 

Let us continue with two simple observations. As long as we only consider the existence of a solution, we can assume the solution $\cF$ is intersection-closed, as if $F,G\in\cF$, then elements of $F\cap G$ know the answer to $F\cap G$ anyway. Another observation is that the family of singletons solves this model if $j\ge n-1$. Indeed, a set $A$ of elements has no information about the other elements, hence they know the defective if and only if he is one of them, or the only element not in the set. This implies $A$ has to have size at least $n-1$. We show that if $i\ge 2$, then this is the only case when Model 4 can be solved.

\begin{thm}\label{model4i2} If $i\ge 2$ and $j\le n-2$, then there is no solution for Model 4.

\end{thm}
\vspace{2mm}

\noindent
The only remaining case is $i=1$. Surprisingly, the solution here depends on divisibility conditions. First we deal with the $j=2$ case. In the following two theorems we prove that a kind of minimal structure should be contained in any solution in this case. 

\begin{thm}\label{Model4i1j21}

If $n \ge 4, \ i=1$ and $j=2$, a Steiner triple system minus a partial matching solves Model 4.

\end{thm}

\begin{thm}\label{steine} 

Let $i=1$ and $j=2$. If $\cF$ is intersection-closed and solves Model 4, then it contains a Steiner triple system on $n$ elements minus a partial matching.

\end{thm}

\noindent
Note that if $i=1$ and $j=2$, then there is a solution for $n=1$ and $n=3$ and there is no solution for $n=2$. So by the previous two theorems and Theorem \ref{steiner} we have:

\begin{cor} Let $i=1$ and $j=2$. There is a solution for Model 4 if and only if $n=6k+1$ or $n=6k+3$.

\end{cor}

\noindent
Now we continue with larger $j$'s.

\begin{thm}\label{model4j3} Let $i=1$. Then we have:

a) if $j\ge 4$ and $n\neq 6$, then there is a solution for Model 4.

b) if $j=3$, $n\neq 6$, $n\neq 6k+2$ and $n\neq 6k+5$ for some integer $k$, then there is a 

\hspace{4mm}solution for Model 4.

\end{thm}

\noindent
The only remaining cases are $i=1$, $j=3$, $n = 6k+2$ or $6k+5$. In every other case we completely characterized the values of $n$ where a solution for Model 4 exists. For our knowledge in the remaining cases see the Remark section.

\subsection{Proofs about Model 4}

Let us start with an easy observation. If $\cF$ is a solution for some $i$ and $j$, then it is a solution for $i'$ and $j'$ with $i'\le i$ and $j'\ge j$.

We will give several constructions that share some basic properties. All the families are linear, meaning that any two query sets intersect in at most one element. There are no two-element query sets. Then an element $x$ can find the defective element only if there are exactly $n-1$ elements contained in sets in $\cF_x$. On the other hand, usually a straightforward case analysis shows that any two (or three, or four) elements together find the defective element, thus in some cases we omit the details.

\begin{proof}[Proof of Theorem \ref{model4i2}] Let us assume indirectly that $\cF$ is a solution. As we remarked earlier we can assume $\cF$ is intersection-closed. Let us remove the singletons from $\cF$ and let $\cF'$ be the resulting family. We claim that $\cF'$ is also intersection-closed. Indeed, if $F,G\in\cF'$ and $| F \cap G|\ge 2$, then their intersection is in $\cF$. On the other side if the intersection would be $\{x\}$, then let $y\in F\setminus \{x\}$, $z\in G\setminus \{x\}$. If $x$ is the defective, $y$ and $z$ together finds that out, which is impossible since $i \ge 2$. Thus $|F\cap G|>1$, hence it is in $\cF'$.

For an element $x \in [n]$ let $F_x := \cap_{x \in F \in \cF'} F$ be the intersection of the sets in $\cF'$ that contain $x$. We have $F_x\in \cF'$. Let $F_y$ be inclusion-wise minimal in $\{F_x : x \in [n]\}$. It has size larger than 1, thus it contains an element $z\neq y$, and we have $F_z\subseteq F_y$ by the definition of $F_z$. Thus we have $F_y=F_z$, which means that $\cF'$ does not separate $y$ and $z$, meaning that they are only separated by singletons (of $\cF$). But then all the other elements (= $[n] \setminus \{y,z\}$) together cannot find which one of $y$ or $z$ is the defective, which is a contradiction as $n \ge 3$ and $j \le n-2$.

\end{proof}

\begin{proof}[Proof of Theorem \ref{Model4i1j21}]

First we show that a Steiner triple system is a solution. Indeed, for any element $a$, if $d$ is the defective with $a\neq d$, then $a$ gets YES answer to the only query $F$ containing both $a$ and $d$. It contains a third element $b$, and $a$ does not know if $b$ or $d$ is the defective as - using that the query family is a Steiner system - $a$ has no more information about $b$. 

\vspace{1mm}
\noindent
On the other hand, let $a' \in [n] \setminus \{a,d\}$. There are two cases. 

\vspace{2mm}

Case 1 : if $a'=b$. 

By $n>3$ there is another query containing $a'$, the answer to that is NO, thus $a'$ knows $a'$ is not defective, similarly $a$ knows about himself that he is not defective, but they both know the defective is in $F$ and so they together can find out it is $d$. 

\vspace{2mm}

Case 2: if $a' \neq b$. 

Then there is a query $F'$ containing both $a'$ and $d$, thus $a$ and $a'$ together know the defective is in $F\cap F'=\{d\}$.

\vspace{2mm}

Let us finish the proof by showing that leaving out a partial matching does not change the information available to the elements. Theorem \ref{steiner} implies $n=6k+1$ or $n=6k+3$ and we have assumed $n\ge 4$, thus we have $n\ge 7$, which means there are at least three queries containing a given element. It is easy to see that if $\{a,b,c\}$ is missing, $a$ knows what the answer to that would be: if $a$ gets exactly one YES answer to the other queries, then it is NO, otherwise it is YES. Indeed, $a$ gets zero YES answer if $b$ or $c$ is the defective, only YES answers (thus at least two of those) if $a$ is the defective, and one YES answer otherwise (for the query that contains $a$ and the defective $d$). 

\end{proof}

Now we prove that a Steiner triple system minus a partial matching is a minimal query family in this case, supposing that the query family is intersection-closed.

\begin{proof}[Proof of Theorem \ref{steine}]
For $a \in [n]$ let $S_a$ be the set of elements that can be defective according to $a$ after getting the answers, and let $S'_a:=S_a\setminus \{d\}$, where $d$ is the defective. Note that $a$ knows $S_a$, but does not know $S'_a$. The property $i=1$ implies $|S_a|\ge 2$ and the property $j=2$ implies $S_a\cap S_b=\{d\}$ if $a\neq b$, $a,b\neq d$. Thus the sets $S'_a$ ($a \in [n], \ a\neq d$) are pairwise disjoint, non-empty sets on an underlying set of size $n-1$. Hence they are singletons as there are $n-1$ of them. This means that for any $a$ there is exactly one element that he cannot distinguish from the defective.

Let us now consider $\cF$. For any $a$, if $d \in [n]\setminus \{a\}$ is considered as the defective, then there is an element $c(a,d) \in [n] \setminus \{a,d\}$ such that $a$ can not distinguish between $d$ and $c(a,d)$. By the remarks above we know that there is exactly one such $c(a,d)$. If there are members of $\cF_a$ that contain both $d$ and $c(a.d)$, then using again the remarks in the previous paragraph, we have that the intersection of them is $\{a,d,c(a,d)\}$, thus it is in $\cF$, as $\cF$ is intersection closed. If there is no such member of $\cF_a$, let us add $\{a,d,c(a,d)\}$ to $\cF$. Let
$$\cF':=\cF \cup\{\{a,d,c(a,d)\} : a \in [n], \ d \in [n] \setminus \{a\}, \ \{a,d,c(a,d)\} \not \in \cF\}.$$

First note that it is impossible that $\{a,d_1,c(a,d_1)\}, \{a,d_2,c(a,d_2)\} \not \in \cF$ with 4 different elements $d_1,d_2,c(a,d_1),c(a,d_2)$ as otherwise $a$ could not distinguish between these elements, which would be a contradiction by the first paragraph of this proof.

Note also that if we add $\{a,b,c\}$ this way because $a$ cannot distinguish $b$ and $c$, then also $b$ cannot distinguish $a$ and $c$ and $c$ cannot distinguish $a$ and $b$. Indeed, let us assume $b$ can distinguish $a$ and $c$, i.e. there is a set $F \in \cF$ that contains $b$ and $c$, but does not contain $a$. There is an element $a'$ such that $b$ cannot distinguish $c$ and $a'$, and thus $\{b,c,a'\}\subseteq F$. Moreover, $\{b,c,a'\}\in\cF$ as it is the intersection of the sets in $\cF$ containing both $b$ and $c$. But this means $a'$ cannot distinguish $b$ and $c$, similarly to $a$, thus they together cannot either, a contradiction. This thought also shows that two sets from $\cF' \setminus \cF$ can not intersect in two elements.
Altogether with the previuos paragraph we have that $\cF' \setminus \cF$ form a partial matching.

Let $\cF_3:=\{F\in\cF': |F|=3$\}. We claim that $\cF_3$ is a Steiner triple system. For any two elements $a,b$ there is a set in $\cF_3$ that contains both as there is an element $c$ such that $a$ cannot distinguish $b$ and $c$; by the above either $\{a,b,c\}\in\cF$ because $\cF$ is closed under intersection, or $\{a,b,c\}$ was added to $\cF$. Moreover, there is exactly one such element $c$, thus exactly one such set.
\end{proof}

\begin{proof}[Proof of Theorem \ref{model4j3}] 

First we note that a PBD-($\{3,4\}$) solves Model 4 with $i=1$, $j=3$ and a PBD-($\{3,4,5\}$) solves Model 4 with $i=1$ and $j=4$. The proof of this statement goes similarly to the proof of Theorem \ref{Model4i1j21}, thus we provide only a sketch here. For any two elements there is a query containing them, and the other elements of that query cannot distinguish the first two. However, any other element can.

The sets of integers $n$ such that there exists such pairwise balanced designs on $n$ elements have been determined by Gronau, Mullin and Pietsch \cite{GMP1995}. They showed that if $n=3k$ or $n=3k+1$ with $n\neq 1,6$, then there exists a PBD-($\{3,4\}$). This proves b). They also showed that if $n\neq 1,2,6,8$, then there exists a PBD-($\{3,4,5\}$). This proves a) except for the case $n=8$. In that case consider the sets $\{1,2,3,4\}, \{1,5,7\}, \{2,5,8\}, \{3,6,8\}, \{4,6,7\}$. One can easily check that these sets solve Model 4.

\end{proof}

\section{Adaptive scenario}

A natural idea is to consider the adaptive versions of these problems. However, the definition of these models are not straightforward. Earlier we assumed the existence of a Questioner only for notational convenience, the elements could come up with the query family in advance. However, in this case it is not clear which one of them should find out the next query in an adaptive algorithm, as they have different information available to them. Here we assume that there is a Questioner who knows all the answers and chooses the next query.

However, there are still two versions of this problem. In the first version the elements know the algorithm, and can use for example the order of the queries to gain information, while in the second version they only receive the family of queries containing them, together with the answers, at the end of the algorithm (thus it is adaptive only for the Questioner). 

Consider for example Model $4$. In the first version Questioner can ask all the singletons, finding the defective this way, and then ask additional queries only to give information to the elements. He wants to share the identity of the defective element as a secret with every $j$-set. He chooses any secret sharing scheme, and to an arbitrary element $x$ he gives its share of the secret by repeating the query $\{x\}$ an appropriate number of times.

On the other hand, we will see that in the second version there is no solution for Model 4 in some cases. In what follows, we only consider the second version.

Note that Questioner can still ask queries only to give information to the elements (just not in a tricky way). For example he can ask queries to find the defective, and then share this information with the elements using further queries. In particular this gives an algorithm of length $\lceil \log_2 n\rceil+2$ for Model 1 and Model 2. After a separating family is asked, Questioner asks the defective $[n]\setminus \{d\}$ and $\{d\}$, if needed.

It is easy to see that Model 3 still cannot be solved. Indeed, let us assume that every answer is YES (unless it would contradict earlier answers). If Questioner finds out that $x$ is the defective, then it is separated from every other element $y$ by a query. The answer to the first such query was YES, thus it contains $x$, and so $x$ knows $y$ is not defective for every $y\neq x$.

Model 3' can be solved using $\lceil \log_2 n\rceil+1$ queries. Questioner starts with the usual halving procedure: first asks a set $F$ of size $\lceil n/2 \rceil$, and then depending on the answer continues recursively with $F$ or $\overline{F}$ as the base set. Then stops when arrives to a set of size less than 6, and asks all but one of the singletons. 

So far there was no difference between the adaptive and non-adaptive versions of the models when considering the existence of a solution. However the situation radically changes with Model 4.

\begin{thm}\label{adaptive}

Let $i=1$. Model 4 can be solved adaptively if and only if $2\le j\le n$ and $n$ is odd, or $3 \le j\le n$ and $n$ is even.

\end{thm}

\begin{proof} 

Let Questioner start with asking the singletons to find the defective element $d$. If $n$ is odd, he partitions the remaining elements into pairs and asks them together with $d$. Then every element $y\neq d$ knows that the defective is either its pair $y'$ or $d$. On the other hand $y$ and $z$ together know it is $d$, as $y'=z'$ cannot happen unless $y=z$. If $n$ is even, one of the parts should contain three of the remaining elements $a,b,c$. Then for example $a$ knows the defective is $d$, $b$ or $c$, and $a$ and $b$ together cannot find the defective, but any three elements can.

Let us now assume $j=2$ and $n$ is even. Let us assume every answer is NO, except if that would lead to a contradiction (note that it still makes sense for Questioner to ask such queries, to help the elements find the defective, as we just saw in the algorithm described above). We claim that in this case there is no solution.

We repeat the beginning of the proof of Theorem \ref{steine}. After the algorithm ends, let $S_a$ be the set of elements that can be defective according to $a$, and let $S'_a=S_a\setminus \{d\}$, where $d$ is the defective. We have $|S_a|\ge 2$ and $S_a\cap S_b=\{d\}$ if $a\neq b$, $a,b\neq d$. Thus the sets $S'_a$, $a\neq d$ are $n-1$ pairwise disjoint, non-empty sets on an underlying set of size $n-1$, thus they are singletons. This means that for any $a$ there is exactly one element that he cannot distinguish from the defective.

Now let us define an auxiliary directed graph on the $n-1$ non-defective elements. Let $y\rightarrow z$ if $y$ cannot distinguish $d$ and $z$, i.e. among the sets that contain $y$, exactly the same sets contain $d$ and $z$. By the above, every out-degree is one in this graph, thus it is the union of directed cycles. Let $y_1,\dots, y_k$ be the vertices of such a cycle $C$ in the cyclic order. If a query contains $d$ and $y_1$, it also contains $y_2$ by the definition of the edges. But then it also contains $y_3$, and so on. It means that the same queries from $\cF_d$ contain the vertices of $C$. Then a vertex in $C$ can distinguish $d$ from other vertices of $C$ only using queries that do not contain $d$. Let us assume $k\ge 3$. Then there is no query containing $y_1$ and $y_2$ and not containing $d$, as $y_1$ cannot distinguish $y_2$ and $d$. However, there must be such a query as $y_2$ can distinguish $d$ and $y_1$ (as $y_1\neq y_3$).

We claim that there is no cycle of length $1$, showing that every cycle is of length $2$, thus $n-1$ is even, finishing the proof. Indeed, a cycle of length $1$ would mean that $y_1$ only received YES answers, thus it only appeared in queries containing $d$. There must be a query that separates $d$ and $y_1$. Consider the first such query. By the above, it cannot contain $y_1$ and avoid $d$, hence it contains $d$ and avoids $y_1$. Thus the answer to it was YES. However, it should have been NO (according to our assumption on the answers), as before that query it was a possibility that $y_1$ is the defective element, thus it would have lead to no contradiction.

\end{proof}

\begin{thm} If Model 4 can be solved adaptively, then $(n-1)\binom{j-1}{i} \ge \binom{n-1}{i}$.

\end{thm}

\begin{proof}
Let us consider again the sets $S_a'$ (defined in the proof of Theorem \ref{steine}) after the end of the algorithm. Let $\cG$ be their family. Let $\cG_k$ be the family of sets that can be written as the intersection of $k$ sets in $\cG$. Then we know that $\emptyset \not\in \cG_i$, but $\cG_j=\{\emptyset\}$. Let us consider the family $\cG'$ of the inclusion-wise minimal non-empty sets, that can be written as the intersection of sets in $\cG$. The members of $\cG'$ are pairwise disjoint, thus there are at most $n-1$ of them. On the other hand each of them can be written as the intersection of at most $j-1$ sets in $\cG$. For every set $G\in \cG'$ let $\cG'_G$ be an inclusion-wise maximal subfamily of $\cG$ such that every member of $\cG'_G$ contains $G$. Then $|\cG'_G|\le j-1$.

Let us take $i$ sets from $\cG$. Their intersection is in $\cG_i$, thus by definition it is a superset of a set $G\in\cG'$. But this can only happen if those $i$ sets are in $\cG'_G$ (otherwise we could add one of those sets to $\cG'_G$, contradicting its maximality). For any $G\in \cG'$ there are at most $\binom{j-1}{i}$ $i$-element subfamilies of $\cG'_G$, and there are at most $n-1$ sets $G\in \cG'$. On the other hand there are $\binom{n-1}{i}$ ways to take $i$ sets from $\cG$.

\end{proof}

This theorem shows that if $i>1$, then $j$ should be large. On the other hand, unlike in the non-adaptive case, $j$ can be smaller than $n-1$. Let us consider the following simple algorithm. Let Questioner ask the singletons first. He finds the defective and then partitions the other $n-1$ elements to $i+1$ sets in a balanced way, and asks all those sets. Any $i$ elements not containing the defective get only NO answers, but there are at least $1+\lfloor (n-1)/(i+1)\rfloor$ elements they do not know anything about. On the other hand if $j>n-1-\lceil (n-1)/(i+1)\rceil$, then $j$ elements without the defective know all the answers to the non-singleton queries, thus they know the defective is the one not appearing in those queries.

\section{Concluding remarks}

We finish this article with some possible directions that can be investigated:

\vspace{2mm}

$\bullet$ In some of the above models we proved that there is a family that solves the model, but did not say anything about its possible size.

$\bullet$ In Model 4 the only remaining case is $i=1$, $j=3$. In this case we only know that a solution exists if $n=6k,6k+1,6k+3,6k+4$. We do not know if it exists for the other values (it does not exist for some small values). 

A simple way to construct a PBD-($\{3,4\}$) is the following. We take a Steiner triple system on a set $X$ of $6k+3$ elements and its partition into $3k+1$ matchings. We take a set $Y$ of $n-6k-3\le 3k+1$ additional elements and a PBD-($\{3,4\}$) on them. Finally, for every element $y\in Y$ we pick one of the matchings, and replace every set $A$ in the matching by $A\cup \{y\}$.

Let us take a family $\cF$ that is a solution for Model 4 with $i=1$, $j=3$ (instead of a PBD-($\{3,4\}$)) on $Y$. Then the resulting family is also a solution. Indeed, an element of $X$ and an element of $Y$ can be distinguished by any element, two elements of $X$ can be distinguished by any element except those two that are in a query with them, and two elements of $Y$ can be distinguished by any elements of $X$ (and all but two elements of $Y$ by our assumption on $\cF$). This argument would give a proof for Theorem \ref{model4j3} without using any characterization of PBDs. 

Additionally, let us assume there is a solution $\cF$ for Model 4 with $i=1$, $j=3$ on $6k_0+2$ elements. Let $k_1\ge 2k_0+1$ and $n=6k_1+3+6k_0+2\ge 18k_0+5$ and take the above construction. Thus we get a solution for any large enough $n=6k+5$. Similarly if we start with a solution on $n=6k_0+5$ elements (or continue with the solution found on $18k_0+5$ elements), we get a solution for large enough $n=6k+2$. Thus a solution for any of the remaining values of $n$ would imply that for every $n$ large enough there is a solution.

$\bullet$ All of the above mentioned models are also interesting in case of $d$ defectives ($d \ge 2$). In a forthcoming paper (\cite{GV2017}) we started such investigations, however a lot of questions remained open.

$\bullet$ In this paper we considered the abstract version of the Model by Tapolcai et al.~\cite{TRHHS2014,TRHGHS2016}. It would be interesting to see if our other models or our methods work with their underlying graph structure.

$\bullet$ Recently there was some interest in the $r$ round (or multi-stage) versions of combinatorial group testing problems (see e.g. \cite{DMT2013,GV2016}). It would be interesting to investigate these models in this context. 
Note that the algorithm provided in Theorem \ref{adaptive} is in fact a 2-round algorithm: in the first round the singletons are asked. With those queries Questioner finds the defective, thus he knows the answer to every later queries (he uses them only to help the elements find the defective). This means whatever algorithm is used afterwards, that can be done in one round. As he gets no new information, there is no point in waiting for the answers.

\subsection*{Acknowledgement}

We would like to thank \'{E}va Hosszu \cite{H2015}, who asked us the first question of the type that was investigated in this article.
We would also like to thank all participants of the Combinatorial Search Seminar at the Alfr\'ed R\'enyi Institute of Mathematics for fruitful discussions.

We also thank the anonymous reviewers for their careful reading of our manuscript and their many insightful comments and suggestions that improved the presentation of our article.

\end{document}